\documentclass[svgnames,11pt]{llncs-href}
\usepackage{color}

 \usepackage{graphicx} 
  \usepackage{amssymb,amstext}

\newcommand{\sm}{\setminus}

\title{Connecting Terminals and 2-Disjoint Connected Subgraphs
}

\author{Jan Arne Telle
\and Yngve Villanger
}
\institute{Department of Informatics, University of Bergen, N-5020 Bergen, Norway. \\
\email{\{telle|yngvev\}@ii.uib.no}
}

\begin{document}

\maketitle

\begin{abstract}
Given a graph $G=(V,E)$ and a set of terminal vertices $T$ we say that
a superset $S$ of $T$ is $T$-connecting if $S$ induces a connected graph,
and $S$ is minimal if no strict subset of $S$ is $T$-connecting.
In this paper we prove that there are at most 
${|V \setminus T| \choose |T|-2} \cdot 3^{\frac{|V \setminus T|}{3}}$ minimal $T$-connecting sets when $|T| \leq n/3$
and that these can be enumerated within a polynomial factor of this bound. 
This generalizes the  algorithm for enumerating all induced paths 
between a pair of vertices, corresponding to the case $|T|=2$.
We apply our enumeration algorithm to solve the {\sc 2-Disjoint Connected Subgraphs} problem
in time $O^*(1.7804^n)$, improving on the recent $O^*(1.933^n)$ algorithm of Cygan et al. 2012 LATIN paper.
\end{abstract}

\section{Introduction}

The listing of all inclusion minimal combinatorial objects satisfying a certain property 
is a standard approach to solving certain $NP$-hard problems exactly. Some examples are 
the algorithms for {\sc Minimum Dominating Set} in time $O^*(1.7159^n)$ \cite{FominGPS08},
for {\sc Feedback Vertex Set} in time $O^*(1.7548^n)$ \cite{FominGPR08}, 
and for {\sc Minimal Separators} in time $O^*(1.6181)$  \cite{FominV12}. 
At the time of their appearance these algorithms were the fastest ones available.

This is an approach that usually requires little in the way of correctness arguments.
For example, in
the {\sc minimum dominating set} problem it is obvious that a dominating set of minimum cardinality is 
also an inclusion minimal dominating set. 
The main task in this approach is to firstly enumerate the inclusion minimal objects, preferably by an algorithm 
whose  runtime is within a polynomial factor of the number of such objects,
 and secondly to provide a good upper bound on the number of objects. 
Probably the most famous example is the polynomial delay enumeration algorithm 
for {\sc Maximal independent set} \cite{JohnsonPY88} 
where there are matching upper and lower bounds on the number of objects \cite{MoonM65}.

Another case with matching upper and lower bounds is the  $O^*(3^{\frac{n}{3}})$ folklore
algorithm enumerating all
induced paths between two fixed vertices $u$ and $v$ in an $n$-vertex graph\footnote{We have not been able to find a proof of this algorithm in the literature.
The graph in Figure 1, with $|R|=1$, shows optimality of the algorithm, up to polynomial factors.}.
In this paper we consider some generalizations of this graph problem.
We first generalize 
to the enumeration of induced paths starting in $v$ and ending in a vertex from a given 
set $R$, with no intermediate vertices in $N(R)$.
The algorithm we give for this generalization will be optimal, up to polynomial factors.
Given a subset of vertices $T$ let us say that a superset $S$ of $T$ is $T$-connecting if $S$ induces a connected graph,
and that $S$ is minimal $T$-connecting if no strict subset of $S$ is $T$-connecting.
Our main generalization is the following enumeration task:

\begin{flushleft}
{\sc Enumeration of Minimal $T$-Connecting Sets}\\
Input: A graph $G=(V,E)$ and a set $T \subseteq V$.\\
Output: All minimal $T$-connecting sets.
\end{flushleft}

Note that for the case $|T|=2$ the minimal $T$-connecting sets are in 1-1 correspondence with the set of induced paths 
between the two vertices of $T$. 
We give an algorithm for {\sc Enumeration of Minimal $T$-Connecting Sets}
with runtime $O^*( {n -|T| \choose |T|-2} \cdot 3^{\frac{n-|T|}{3}})$ where $|T| \leq n/3$. 
For $|T| > n/3$ a trivial $O^*(2^{n-|T|})$ brute force enumeration can be used. 
We apply this enumeration algorithm to solve the following problem:

\begin{flushleft}
{\sc 2-Disjoint Connected Subgraphs}\\
Input: A connected graph $G=(V,E)$ and two disjoint subsets of terminal vertices $Z_1,Z_2 \subseteq V$.\\
Question: Does there exist a partition $A_1,A_2$ of $V$, with 
$Z_1 \subseteq A_1, Z_2 \subseteq A_2$ and $G[A_1]$, $G[A_2]$ both connected?
\end{flushleft}

The general version of this problem with an arbitrary number of sets was used as one of the tools
in the result of Robertson and Seymour showing that {\sc Minor containment} can be solved in
polynomial time for every fixed pattern graph $H$ \cite{RobertsonS95b}.
We require the input graph to be connected since otherwise it is easy to reduce the problem to
a connected component. 

Let us look at some previous work on this problem.
Motivated by an application in computational geometry, Gray et al \cite{GrayKLS12} showed that 
{\sc 2-Disjoint Connected Subgraphs} is NP-complete on planar graphs. van't Hof et al \cite{HofPW09} showed that 
on general graphs it is NP-complete even when $|Z_1|=2$ and also that it remains NP-complete on $P_5$-free graphs 
but is polynomial-time solvable on $P_4$-free graphs.
Notice that the naive brute-force algorithm that tries all 2-partitions of non-terminal vertices
runs in time $O(2^kn^{O(1)})$, where $k=n-|Z_1 \cup Z_2|$.
This shows that {\sc 2-Disjoint Connected Subgraphs} is fixed-parameter tractable when parameterizing by the number of non-terminals.
However, Cygan et al \cite{CyganPPW12} show that
breaking this $O^*(2^k)$ barrier for the number $k$ of non-terminals would contradict the Strong Exponential Time Hypothesis, and
that a polynomial kernel for this parameterization would imply $NP \subseteq coNP/poly$.
Paulusma and van Rooij \cite{PaulusmaR11} gave
an algorithm with runtime $O^*(1.2051^n)$ for $P_6$-free graphs and asked whether it was possible to solve the
problem in general graphs faster than $O(2^n n^{O(1)})$.
This question was recently answered affirmatively by Cygan et al \cite{CyganPPW12}
who gave an algorithm for {\sc 2-Disjoint Connected Subgraphs} on general graphs, based on the branch and reduce technique, with runtime $O^*(1.933^n)$.

Our algorithm for {\sc 2-Disjoint Connected Subgraphs} on general graphs will be based on {\sc Enumeration of Minimal $T$-Connecting Sets} and have runtime $O^*(1.7804^n)$.

Our paper is organized as follows. In Section 2 we give 
the main definitions. 
In Section 3 we address the enumeration of induced paths starting in $v$ and ending in a vertex from a given 
set $R$, with no intermediate vertices in $R$.
In Section 4 we
give an algorithm for 
{\sc Enumeration of Minimal $T$-Connecting Sets}.
In Section 5 we apply this enumeration algorithm to solve 
the {\sc 2-Disjoint Connected Subgraphs} problem.
We end in Section 6 with some questions.

\section{Definitions}

We deal with simple undirected graphs and use standard terminology. 
For a graph $G=(V,E)$ and $S \subseteq V$ we denote by $G[S]$ the graph induced by $S$.
An induced subgraph $G[S]$ for $S \subset V$ is called connected if any pair of vertices of $S$ 
are connected by a path in $G[S]$.
We may also denote the vertex set of a graph $G$ by $V(G)$.
We denote by
$N[S]$ the set of vertices that are in $S$ or have a neighbor in $S$, and let $N(S)=N[S] \setminus S$.

A path $P$ of a graph $G$ is a sequence of vertices $(v_1,v_2,\ldots,v_q)$ such that 
$v_jv_{j+1} \in E$ for $1 \leq j < q$, and the path is called induced if $G[\{v_1,v_2,\ldots,v_q\}]$ 
has no other edges.
A subpath of $P$ is of the form $(v_1,v_2,\ldots,v_i)$ for some $i \leq q$.

Contracting an edge $uv$ into vertex $v$ in a graph $G$ is defined as the operation of 
adding, for every vertex $w \in N(u) \sm N[v]$, the edge $vw$ to $G$ if it is not already present, 
and then deleting $u$ and all edges incident to $u$.
Notice that a graph is connected after the contraction operation if and only if it was connected before the 
contraction operation. 


Given a graph $G = (V,E)$, a vertex set $T \subset V$, a vertex $v_1 \in V \sm T$, and 
an induced path $P=(v_1, v_2,...,v_q)$ in $G[V \sm T]$, 
we define the {\it branch depth} of path $P$ to be $$b(P) = |N[\{v_1,v_2,...,v_{q-1}\}]|-1.$$

\section{Induced paths from a vertex to a set of vertices}

It is folklore knowledge that the set of induced paths between a pair of vertices in an $n$-vertex graph can be 
enumerated in $O^*(3^{\frac{n}{3}})$ time. 
We have not been able to find a written proof of this in the literature.
In the following theorem the induced paths between a pair of vertices is a special case, thus providing a generalization of a well known result. 

\begin{theorem}
\label{new}
Given a graph $G=(V,E)$, 
a vertex $v \in V$ and 
$R \subseteq V \sm N[v]$, we can enumerate all induced paths from $v$ to a vertex of $N(R)$, with no intermediate vertex in $N[R]$,
in time  $O^*(3^{\frac{|V \setminus R|}{3}})$. 
\end{theorem}

\begin{figure}\label{fig1}
\centering
\includegraphics[scale=0.4]{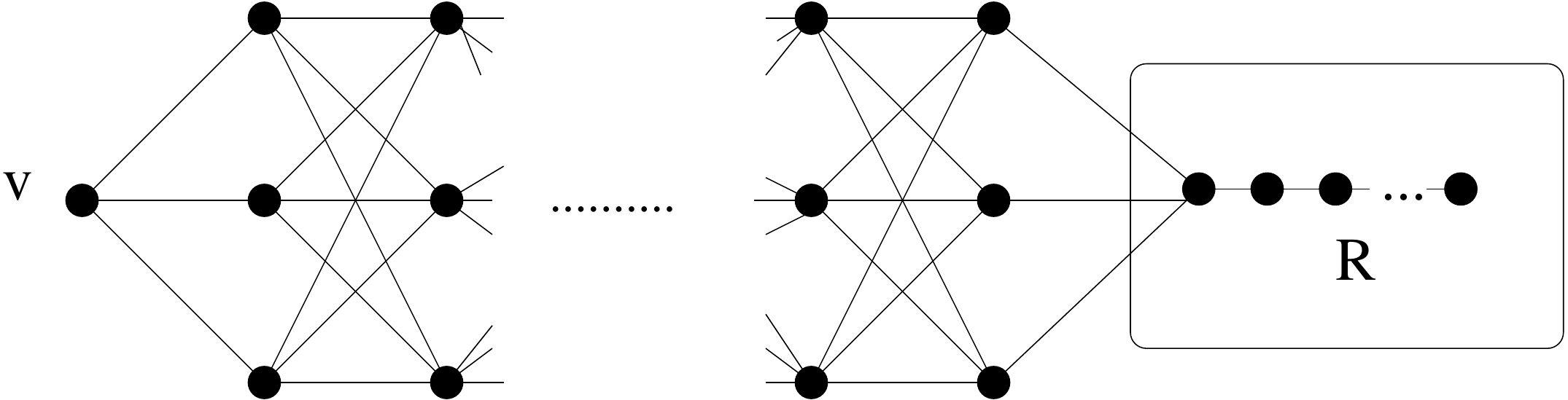}
\caption{The number of induced paths between vertex v and a vertex of $N(R)$ (the rightmost column of 3 vertices) is 
$3^{\frac{n-|R|-1}{3}}$. Since each such path $P$ has branch depth $b(P)=n-|R|-1$ this graph shows tightness of 
Lemma \ref{pr:basic_path} when $b(P)$ is a multiple of 3. If $b(P)=3i+1$ then replace one column of 3 vertices by 4 vertices
and if $b(P)=3i+2$ add a new column of 2 vertices.
$R$ induces a connected graph so the number of minimal $R \cup \{v\}$-connecting sets is also $3^{\frac{n-|R|-1}{3}}$.}
\end{figure}

We actually want the paths from $v$ to $R$, but since these paths must have the second-to-last vertex in $N(R)$ we state the result as above. 
Theorem \ref{new} will follow from Lemma \ref{pr:basic_path}, which is stated in terms of branch depth of paths in order to be used for the branching algorithm in the next section. Since the branch depth of each induced path from $v$ to $N(R)$, with no intermediate vertex in $N(R)$, is at most $|V \setminus R|-1$, Theorem \ref{new} will follow from Lemma \ref{pr:basic_path} below and is tight up to polynomial factors, see Figure 1.
We start with a combinatorial lemma.

\begin{lemma}\label{comb}
Fix a non-negative integer $t$ and let $T$ be a rooted tree where any root-to-leaf path $v_1,v_2,...,v_q$ has
$\Sigma_{1 \leq i \leq q} \mbox{ } c(v_i) \leq t$, with $c(v)$ the number of children of node $v$.
The maximum number of leaves that $T$ can have is $l(t)$ with $l(1)=1$ and for $t \neq 1$
$$
l(t) = \left\{
\begin{array}{ll} 3^i & \mbox{if $t=3i$,} \\
4 \cdot 3^{i-1} & \mbox{if $t=3i+1$,}\\
2 \cdot 3^i & \mbox{if $t=3i+2$}.\\
\end{array}
\right.
$$
\end{lemma}

\begin{proof}
We first show that for any $t$ there is a tree $U_t$ achieving the maximum, where all nodes at the same level have the same number of children.
For any $t$ let $T_t$ be any rooted tree achieving the maximum. Define $r(t)$ as the number of children of the root of $T_t$. 
In the tree $U_t$ all nodes at level $i \geq 1$ will have $u(i)$ children, with $u(i)$ defined level-by-level as follows.
The root of $U_t$ will have the same number of children as the root of $T_t$, in other words we define $u(1)=r(t)$. 
The sum of the number of children of nodes on any path from a child of the root of $U_t$ to a leaf of $U_t$ should be $t-u(1)$, thus
nodes at level 2 of $U_t$ should have the same number of children as the root of $T_{t-u(1)}$, in other words
we define $u(2)=r(t-u(1))$. Continuing like this we get that in general
$u(i)=r(t - \Sigma_{1 \leq j < i} \mbox{ } u(j))$. 
By induction on $t$ it follows that $U_t$ has as many leaves as $T_t$ and any root-to-leaf path has $t$ children.

Assume $U_t$ has $p$ levels. We then have that $u(1)+u(2)+...+u(p-1)=t$ and that $u(1)\cdot u(2)...\cdot u(p-1)$, the number of leaves
of $U_t$, is maximized. Since the product of these integers is maximized we can assume that we have no integer $x \geq 4$ among them since then we could replace $x$ by $2\cdot (x-2) \geq x$ which does not decrease the product nor changes the sum of the integers. 
Also, if $2$ appears then it appears at most twice since we could replace $2\cdot 2\cdot 2$ by $3\cdot 3 > 2\cdot 2\cdot 2$. This implies that the number of leaves in $U_t$ is $l(t)$ as stated in the Lemma.
\end{proof}

Note that $l(t)$ is the maximum number of maximal independent sets in a graph on $t$ vertices \cite{MoonM65,FomKra}. For the connection to 
the largest integer which is
the product of positive integers with sum $t$ see e.g.
 \cite{Vatter:2011:MIS}.

\begin{lemma}\label{pr:basic_path}
Given a graph $G=(V,E)$, 
a vertex $v_1 \in V$, 
$R \subseteq V \sm N[v_1]$, and  
an integer $t$.
Then there exist at most $l(t)$ induced paths $P=(v_1, v_2,...,v_q)$ in $G$ such that 
\begin{itemize}
 \item $b(P) \leq t$, 
 \item $v_i \not\in N[R]$ for $1 \leq i \leq q-1$, and 
 \item $v_q \in N(R)$.
\end{itemize}
Furthermore all these paths can be enumerated in $O^*(3^{\frac{t}{3}})$ time. 
\end{lemma}

\begin{proof}
The enumeration algorithm will be a standard backtracking algorithm starting in 
$v_1$ that checks all choices. 
At the first step the choices for $v_2$ are the vertices in $N(v_1)$.
In general, when we have a subpath $P=(v_1, v_2,...,v_i)$, if $v_i \not \in N(R)$
the choices for $v_{i+1}$ are 
the vertices in $N(v_i) \setminus N[\{v_1,v_2,...,v_{i-1}\}]$. 
If $v_i \in N(R)$ then we have a leaf in the tree $T$ of all possible choices.
Thus,
in the rooted tree $T$ of all possible choices, if we label the nodes of $T$ with the vertex chosen,
the set of paths from the 
root to a leaf in $T$ will be in 1-1 correspondence with the set of paths satisfying the statement in the Lemma 
without any bound on branch depth.

Consider such a path $P=(v_1, v_2,...,v_q)$. By definition the branch depth of $P$ is
$b(P)= |N[\{v_1,v_2,...,v_{q-1}\}]|-1$. Consider the leaf-to-root path $P_T$ in $T$ corresponding to $P$.
For any $1 \leq i < q-1$ the children of the node in $P_T$ labelled $v_{i}$ have labels $N(v_i) \setminus N[\{v_1,v_2,...,v_{i-1}\}]$, and the node labelled $v_q$ is a leaf. Thus the children of all nodes of $P_T$ have distinct labels and the union of all these labels is
$N[\{v_1,v_2,...,v_{q-1}\}]\setminus \{v_1\}$. Thus
the sum of the number of children over all nodes on $P_T$
is exactly $b(P)$.

Consider any rooted tree $T$ having the property that for any root-to-leaf path the sum of the number of children of all nodes on this path is at most $t$. Lemma \ref{comb} bounds the number of leaves in such a tree to $l(t)$. 
By the above observations, and the fact that
$l(t) \leq 3^{t/3}$ since $2 \leq 3^{2/3}$ and $4 \leq 3 ^{4/3}$, this proves the Lemma.

\end{proof}

This enumeration algorithm is optimal to within polynomial factors, see Figure 1.
%

\section{Enumeration of Minimal $T$-Connecting Sets}

Theorem \ref{new} with $|R|=1$ can be viewed as an enumeration of all minimal $T$-connecting sets when $T=\{u,v\}$. We now generalize this approach to an arbitrary terminal set $T$ by a branching algorithm.
The following observation will be used to simplify our branching algorithm.

\begin{lemma}\label{le:contraction}
Given $G=(V,E)$, $T \subseteq V$, and two vertices $u,v \in T$ such that $uv \in E$. 
Let $G'$ be the graph obtained by contracting edge $uv$ into $v$. 
Then there is a one to one mapping between Minimal $T$-Connecting Sets in $G$ and 
Minimal $T \sm \{u\}$-Connecting Sets in $G'$. 
\end{lemma}
\begin{proof}
For every Minimal $T$-Connecting Set $S$ in $G$ we can contract edge $uv$ and obtain 
a minimal $T \sm \{u\}$-Connecting Sets $S' = S \sm \{u\}$ in $G'$. 
For every minimal $T \sm \{u\}$-Connecting Sets $S'$ in $G'$ we can observe that $G[S' \cup \{u\}]$ is a 
$T$-Connecting Set in $G$ and it is also minimal as $u \in T$.
\end{proof}

Consider {\bf Algorithm Main Enumeration}. It
will solve {\sc Enumeration of Minimal $T$-Connecting Sets} for any graph $G=(V,E)$ and $T \subseteq V$.
Let us first give the informal intuition for the algorithm. We fix a vertex $u \in T$ and using the algorithm of Lemma \ref{pr:basic_path} we find all induced paths from $u$ to $N(T \setminus \{u\})$. For each of these paths $P$ we again call the
algorithm  of Lemma \ref{pr:basic_path}, but now on the graph $G'$ where the path $P$ together with the vertices of $T$ that $P$ has in its neighborhood, is contracted into $u$. The path we find in $G'$ will in $G$ start at some vertex of $P$ or a neighbor in $T$ and we see that we start forming a tree of paths.
We carry on recursively in this way until the collection of paths spans all of $T$,
note however that the vertices of these paths may induce a graph containing cycles.
To avoid repeating work we label vertices by a total order and use this ordering to guide the recursive calls.

\begin{figure*}[t]
\normalsize
\begin{tabbing}

xxx\=xxx\=xxx\=xxx\=xxx\=xxx\=xxx\=xxx\=xxx\= \kill
{\bf Algorithm Main Enumeration}\\
{\bf Input:} A graph $G = (V,E)$ and terminal set $T \subseteq V$\\
{\bf Output:} A family of sets containing all Minimal $T$-Connecting Sets\\
{\bf begin}\\
\> {\bf assign} each vertex a unique $label$ between 1 and $|V|$\\
\>{\bf choose} $u \in T$\\
\>{\bf MCS}$(\emptyset, \emptyset)$\\
{\bf end}\\
\\
{\bf Procedure MCS$(C,X)$}\\
{\bf Parameter $C$}: vertex set used to connect $T$\\
{\bf Parameter $X$}: vertices not to explore in this call \\
{\bf begin}\\
{\bf if} $G[T \cup C]$ is connected then {\bf output} $T \cup C$\\
{\bf else}\\
\>{\bf set} $C_u \supseteq C$ as vertex set of connected component of $G[T \cup C]$ containing $u$\\
\>{\bf set} $T'=T \setminus C_u$ i.e. the terminals not yet connected to $u$ by $C$ \\
\>{\bf set} $G'$ to be graph obtained from $G$ by contracting edges of $G[C_u]$ to $u$\\
\>{\bf call} the algorithm of Lemma \ref{pr:basic_path} on $G'[V(G') \sm X]$ with $v_1=u$ and $R=T'$\\
\>{\bf for} every path $P = (v_1,v_2,\ldots,v_q)$ output by that call\\ 
\>\> {\bf MCS}$(C \cup \{v_2,\ldots,v_q\}, \mbox{ } X \cup \{w \in N(C_u): label(w) < label(v_2)\})$\\
\>{\bf end-for}\\
{\bf end}

\end{tabbing}
\label{alg:enum}
\end{figure*}

\begin{lemma}\label{le:connected}
Given $G=(V,E)$, $T \subseteq V$ and $|T| \leq n/3$ Algorithm Main Enumeration will:
\begin{enumerate}
 \item output every Minimal $T$-Connecting Set of $G$,
 \item output, for any integer $r \in [0..|V \setminus T|]$, at most ${|V \sm T| \choose |T|-2} \cdot 3^{r/3}$ vertex sets 
 $S \supseteq T$ such that $|N[S] \sm T| \leq r$, and 
 \item run in $O^*( {|V \setminus T| \choose |T|-2} \cdot 3^{|V \setminus T|/3})$ time.
\end{enumerate}
\end{lemma}

\begin{proof}
1.)
Let us first argue that every Minimal $T$-connecting vertex set is output by the algorithm. 
In the case where $|T| \leq 1$ the single vertex set $T$ is output by the algorithm.
In the remaining cases $|T|> 1$.

Let $S$ be a minimal $T$-connecting vertex set.
Our goal will be to show that 
there will be a call MCS$(C,X)$ performed by the algorithm in which $T \cup C = S$.
Initially $C= X = \emptyset$ so we trivially have $T \cup C \subseteq S$  and $X \cap S = \emptyset$.
Consider a call MCS$(C,X)$ where we have  $|T \cup C|$ maximized 
under the constraint $T \cup C \subseteq S$ and $S \cap X = \emptyset$.
We show by contradiction that $T \cup C = S$ for this call MCS$(C,X)$.
Assume, by sake of contradiction, that there is a terminal vertex not in $C_u$, 
i.e. not in the component of $G[T \cup C]$
containing $u$, i.e. that $T' \neq \emptyset$. 
%
Let $v_2$ be the lowest numbered vertex of $(N(C_u) \cap S) \setminus X$. 
As $S$ is minimal we have that $G[S]$ is connected but $G[S \sm \{v_2\}]$ is not connected. 
By the minimality of $S$ we have that $G'[S']$ is connected but $G'[S' \sm \{v_2\}]$ is not connected 
for $S' = (S \sm C_u) \cup \{u\}$.
Vertex $v_2$ is not a vertex of $C \cup T$ and as $S$ is minimal we have that each connected component of $G'[S' \sm \{v_2\}]$ contains a vertex of $T'$. If this was not the case, this component could simply be removed from $S$ without changing the connectivity between 
vertices of $T$.
Let $B$ be a connected component of $G'[S' \sm \{v_2\}]$ not containing $u$.
By the previous arguments $B$ contains a vertex of $T'$.
Therefore the call of the algorithm in Lemma \ref{pr:basic_path} on graph $G'$ with $R=T'$ will find a path 
$P=(u,v_2,\ldots,v_q)$ with all vertices in $S$ and with $v_q$ a neighbor of a vertex of $T'$ in $B$ and containing only vertex $v_2$ from $N(C_u)$.
This would lead to a recursive call where $C$ would be updated to $C \cup \{v_2, \ldots, v_q\} \subseteq S$,
and to $X$ there would not be added any vertices of $S$ as $v_2$ had lowest label among all vertices in $(N(C_u) \cap S) \setminus X$,
contradicting the maximality of $|T \cup C|$ under the constraint 
$T \cup C \subseteq S$ and 
$S \cap X = \emptyset$.

2.)
We bound the number of recursive calls in the algorithm and thus also the number of vertex sets that is output. 
Our objective will be to prove that the number of recursive calls MCS$(C,X)$ where 
$r = |N[C_u] \sm T|$ and 
$p$ is the number of times a path is added to $C$, 
is at most ${|X| +p \choose p-1} \cdot 3^{r/3}$. Note that $p$ is equal to the depth of the recursion.
Let $x=|X|$.
As $X \subset N(C_u)$ by the construction of the algorithm,
$p \leq |T|-1$, and at least one vertex is added to $C$ for each found path so $p \leq |C|$, 
we have that $x + p \leq |N[C_u] \sm T| = r$. 
Given that $|T|\leq n/3$ and thus $|V \sm T| \geq 2|T|$ it is clear that ${|V \sm T| \choose |T|-2} \geq {x+p \choose p-1}$ and the claim of the lemma follows. 

The proof will be by induction on $\ell = x+p$.
We assume without loss of generality that $|T| \geq 2$.
The first call is MCS$(\emptyset, \emptyset)$ in which case $p = 0$, 
and this is in fact the only call where $x+p \leq 0$.
The execution of MCS$(\emptyset, \emptyset)$ will call the algorithm of Lemma \ref{pr:basic_path} on $G'$ with $v_1=u$ and $R=T'$ and make a recursive call MCS$(C,X)$ for each path $P$ output by the algorithm of Lemma \ref{pr:basic_path}.
Consider such a call MCS$(C,X)$ originating from path $P$. 
This call will have $x=|X| \geq 0$,  $p = 1$, and it will have $r=|N[C_u] \sm T| \geq b(P)$. 
The number of paths $P$ with $b(P) \leq r$ output by the algorithm of Lemma \ref{pr:basic_path}
applied to the execution of MCS$(\emptyset, \emptyset)$ on $G'$ with $v_1=u$ and $R=T'$
is at most $3^{r/3}$.
Since $3^{r/3} \leq {x + p \choose p-1} \cdot 3^{r/3}$ for $p=1$ we have just established the 
base case $\ell = x+p \leq 1$ this also covers all cases where $p \leq 1$ in our induction.

In the induction step we consider the case where $\ell = x+p \geq 2$ and $p > 1$. 
Let MCS$(C',X')$ be a call and let 
$x' = |X'|$, 
$r' = |N[C_u'] \sm T|$, and 
$p'$ be the number of paths added, or equivalently the depth of the recursion.
By the induction hypothesis we assume that the bound holds for 
the number of calls MCS$(C',X')$ where $x' +p' \leq x +p -1$. 

Every call MCS$(C,X)$ where $x+p = \ell$ is created by a call MCS$(C',X')$ 
and a path $P = (v_1,v_2,\ldots,v_q)$ such that 
$C = C' \cup \{v_2,\ldots,v_q\}$,
$p' = p-1$, and 
$X= X' \cup \{w \in N(C_u'): label(w) < label(v_2)\}$.
As each vertex from $N(C_u') \sm X'$ chosen as $v_2$ will create a unique size 
of the set $X= X' \cup \{w \in N(C_u'): label(w) < label(v_2)\}$ for the next recursive call 
there is at most one choice for $v_2$ 
starting from a fixed MCS$(C',X')$ when it should lead to a recursive call MCS$(C,X)$
where $x+p = \ell$. 
However, there are choices for the sub-path vertices $(v_3, \ldots , v_q)$, but these vertices can be chosen only among 
$V \sm (N[C_u'] \cup T)$, since $v_2$ is fixed in $N(C_u')$ and the path 
$P$ is induced. Note that any such sub-path has branch-depth at most $|N[C_u] \sm (N[C_u'] \cup T)|$.
We can use Lemma \ref{pr:basic_path} to bound the number of such sub-paths, as follows.
By applying Lemma \ref{pr:basic_path} to the graph we get from $G[V \sm (N(C_u') \sm \{v_2\}) ]$ by contracting $C_u' \cup \{v_2\}$ to $u=v_1$ and with $R=T \setminus C_u'$ we deduce that the number of such sub-paths is at most $3^{(|N[C_u] \sm (N[C_u'] \cup T)|)/3}$.

This means that the number of calls MCS$(C,X)$ where $x+p= \ell$ is at most the 
number of calls MCS$(C',X')$ where 
$C' \subseteq C$, 
$X' \subseteq X$ thus $x' \leq x$, and  
$p' =  p-1$, times $3^{(|N[C_u] \sm (N[C_u'] \cup T)|)/3}$.
By the induction hypothesis we have that the 
number of calls  MCS$(C',X')$ where $x'+ p' < \ell$ is at most 
${x' + p-1 \choose p-1-1} \cdot 3^{|N[C_u'] \sm T|/3}$.
Multiplying these two factors we get ${x' + (p-1) \choose (p-1)-1} \cdot 3^{|N[C_u'] \sm T|/3} \cdot 3^{(|N[C_u] \sm (N[C_u'] \cup T)|)/3}$
which can be simplified to ${x' + (p-1) \choose (p-1) -1} \cdot 3^{(|N[C_u] \sm T|)/3}$.

Thus it remains to bound the number of calls MCS$(C',X')$ that can make a new recursive call 
MCS$(C,X)$ where $x+p = \ell$ to be at most ${x + p \choose p-1}$. 
We know that each call MCS$(C',X')$ can only make calls where $x+p = \ell$ when it uses the unique vertex $v_2 \in N(C_u') \sm X$ 
as the second vertex of the path. 
Thus it suffices to count these calls, and let $y$ be the number of such calls. 
We have that
\[
 y \leq \sum_{i = 0}^{x} {i + p-1 \choose p-1-1} \cdot 3^{(|N[C_u] \sm T|)/3}
\]

Using the standard observation that $\sum_{k = 0}^n {k \choose m} = {n+1 \choose m+1}$ we can conclude that 
$y \leq {x + p \choose p-1} \cdot 3^{(|N[C_u] \sm T|)/3}$ and the proof is completed. 

\bigskip
3.)
In the previous claim we bounded the number of recursive calls in the algorithm to ${|V \sm T| \choose |T|-2} \cdot 3^{r/3}$ 
vertex sets $S \supseteq T$ such that $|N[S] \sm T| \leq r$ and $|T| \leq n/3$, and as Lemma \ref{pr:basic_path} ensures that all paths in a 
single call can be enumerated within a polynomial delay it follows that the polynomial bound holds. 
\end{proof}

Using Lemma \ref{le:connected} we can make the following conclusion. 

\begin{theorem}\label{th:connected}
For an $n$ vertex graph $G =(V,E)$ and a terminal set $T \subseteq V$ where $|T| \leq n/3$
there is at most ${n - |T| \choose |T|-2} \cdot 3^{(n-|T|)/3}$ minimal $T$-connecting vertex sets 
and these can be enumerated in $O^*({n - |T| \choose |T|-2} \cdot 3^{(n-|T|)/3})$ time. 
\end{theorem}

\section{The 2-Disjoint Connected Subgraphs problem}

Let us now use Theorem \ref{th:connected} to solve the 2-Disjoint Connected Subgraphs problem.
Recall that the problem is defined as follows:

\begin{flushleft}
{\sc 2-Disjoint Connected Subgraphs}\\
Input: A connected graph $G=(V,E)$ and two disjoint subsets of vertices $Z_1,Z_2 \subseteq V$.\\
Question: Does there exist two disjoint subsets $A_1,A_2$ of $V$, with 
$Z_1 \subseteq A_1, Z_2 \subseteq A_2$ and $G[A_1]$, $G[A_2]$ both connected?
\end{flushleft}

\begin{theorem}
There exists a polynomial space algorithm that solves the {\sc 2-Disjoint Connected Subgraphs} problem in $O^*(1.7804^n)$ time.
\end{theorem}

\begin{proof}
Let us assume without loss of generality that $|Z_1| \leq |Z_2|$ and let $\alpha = |Z_1|/n$; note that
 $0 < \alpha \leq 0.5$.
The algorithm has a first stage that finds a list of potential candidates for $A_1$ and a second stage that checks each
candidate to see if it can be used as a solution. In the first stage we choose between two different strategies depending on whether or not $\alpha > 0.0839$.

Consider first the case where $\alpha \leq 0.0839$.
Vertices of $Z_2$ are of no use when searching for a potential set $A_1$ so it suffices to consider the graph $G[V \sm Z_2]$.
By the algorithm for {\sc Enumeration of Minimal $T$-Connecting Sets} of Theorem~\ref{th:connected} we know that  for $|Z_1| \leq n/3$ in the graph $G[V \sm Z_2]$ all minimal $Z_1$-connecting sets can be enumerated in 
$O^*( {n - |Z_1| -|Z_2| \choose |Z_1|-2} \cdot 3^{(n-|Z_1|-|Z_2|)/3})$ time.
As $|Z_1| \leq |Z_2|$ it is clear that $\alpha n \leq |Z_2|$.
The number $|Z_2|$ only contributes negatively so we can observe that
\[ {n - |Z_1| -|Z_2| \choose |Z_1|-2} \cdot 3^{(n-|Z_1|-|Z_2|)/3} \leq
 {(1-2\alpha)n \choose \alpha n -2} \cdot 3^{(1-2\alpha)n/3}.
\]

By using the Stirling approximation we know that 
${(1-2\alpha)n \choose \alpha n -2}$ is 
$O^*((\frac{\beta^\beta}{\alpha^\alpha \cdot (\beta-\alpha)^{(\beta - \alpha)}})^n)$ 
where $\beta = (1-2\alpha)$.
It is not hard to verify that the maximum value of 
${(1-2\alpha)n \choose \alpha n -2} \cdot 3^{(1-2\alpha)n/3}$ for $0 < \alpha \leq 0.0839$ occurs when $\alpha = 0.0839$ and that 
${(1-2\alpha)n \choose \alpha n -2} \cdot 3^{(1-2\alpha)n/3} \leq 1.7804^n$ for $\alpha = 0.0839$. Thus, we can conclude that when $\alpha \leq 0.0839$
a list of all
minimal $Z_1$-connecting sets can be found in time $O^*(1.7804^n)$.

Consider now the case where $\alpha > 0.0839$.
In this case the algorithm simply loops over all subsets of $V \sm (Z_1 \cup Z_2)$ to list every vertex subset $A \subseteq (V \sm Z_2)$ where $Z_1 \subseteq A$.
As $\alpha > 0.078$ and $\alpha n = |Z_1| \leq |Z_2|$ we get that the number of such subsets is at most $2^{n-2\alpha} \leq 1.7804^n$ and
they can be found in $O^*(1.7804^n)$ time.

For the second stage of the algorithm, for every listed set $A$,  the algorithm tests if vertices of $Z_2$ are contained in the same connected 
component of $G \sm A$ and if
so the algorithm returns the solution with $A_1 = A$ and $A_2$ being the vertices of the connected component of $G \sm A$ containing $Z_2$.
This is clearly a solution to the problem. Conversely, if there is a solution $A_1,A_2$ to the problem, then there is clearly one where
$A_1$ is a minimal $Z_1$-connecting set.

Finally, observe that the algorithm uses polynomial space as a simple branching algorithm is used for both cases.
\end{proof} 
%
%
%
%

\section*{Conclusion}

The graph in Figure 1 shows that our algorithm for {\sc Enumeration of Minimal $T$-Connecting Sets} given by Theorem \ref{th:connected}
is optimal, up to polynomial factors, for the case $|T|=2$. 
Is the algorithm optimal, up to polynomial factors, also for larger $T$, let us say $|T| \leq 0.1n$?

Let us remark that our algorithm for {\sc Enumeration of Minimal $T$-Connecting Sets} can be used to give a 
$O^*({|V \setminus T| \choose |T|-2} \cdot 3^{\frac{|V \setminus T|}{3}})$ algorithm for {\sc Steiner Tree with unit weights} on terminal vertices $T$. This is upper bounded by 
$O^*(1.8778^ n)$ when balanced with the standard brute force search, but will not beat the fastest algorithm for this problem, which is by Nederlof \cite{Nederlof09} and has runtime $O^*(1.3533^n)$ using polynomial space.

The algorithm given in this paper for {\sc Enumeration of Minimal $T$-Connecting Sets} may have more applications in the future, apart from {\sc 2-Disjoint Connected Subgraphs}, in particular for problems where the enumeration of all solutions is required.

\bibliographystyle{plain}
\bibliography{2dis}

\end{document}